\documentclass[conference]{IEEEtran}
\usepackage{amsmath,amsthm,amsfonts,amssymb,authblk,multirow,multicol}
\usepackage{graphicx,array,cite,caption,mathrsfs}
\usepackage{epstopdf,placeins}
\usepackage{mathtools}
\usepackage{tikz}
\usetikzlibrary{shapes,arrows}
\tikzstyle{line}=[draw]
\usepackage{bigstrut}

\usepackage[utf8]{inputenc}

\theoremstyle{remark}
\newtheorem{thm}{Theorem}

\newtheorem{lem}{Lemma}

\newtheorem{cor}{Corollary}
\newtheorem{defn}{Definition}
\newtheorem{exmp}{Example}

\newtheorem{rem}{Remark}

\DeclareRobustCommand{\rchi}{{\mathpalette\irchi\relax}}
\newcommand{\irchi}[2]{\raisebox{\depth}{$#1\chi$}}
\ifCLASSINFOpdf
\else
\fi

\usepackage[utf8]{inputenc}

\renewcommand{\IEEEbibitemsep}{0pt plus 0.5pt}
\makeatletter
\IEEEtriggercmd{\reset@font\normalfont\fontsize{7.9pt}{8.40pt}\selectfont}
\makeatother
\IEEEtriggeratref{1}
\IEEEaftertitletext{\vspace{-10\baselineskip}}

\usepackage{etoolbox}
\makeatletter
\patchcmd{\@maketitle}
{\addvspace{0.5\baselineskip}\egroup}
{\addvspace{-.05\baselineskip}\egroup}
{}
{}
\makeatother
\begin{document}
	\title{Optimal Weakly Secure Linear Codes for Some Classes of the Two-Sender Index Coding Problem}
	\author{\vspace{-.3in}\\Chinmayananda Arunachala, and B. Sundar Rajan\\
		Email: \{chinmayanand,bsrajan\}@iisc.ac.in}

\maketitle
\begin{abstract}
	The two-sender unicast index coding problem is the most fundamental multi-sender index coding problem. The two senders collectively cater to the demands of all the receivers, by taking advantage of the knowledge of their side-information. Each receiver demands a unique message and has some side-information. Weakly secure index coding problem is a practical version of the index coding problem in the presence of an eavesdropper. The eavesdropper can not gain any information about the messages he does not have, by listening to the senders' transmissions. We provide constructions of weakly secure linear codes for different classes of the two-sender unicast index coding problem, using those of its sub-problems. The constructions are valid only if such codes exist for all the sub-problems under consideration. We identify some classes of the two-sender problem, where the constructions provide optimal weakly secure linear index codes.
\end{abstract}

\section{Introduction}
The index coding problem (ICP) introduced in \cite{BK} consists of a single-sender which caters to the demands of all the receivers by availing the knowledge of their side-information. The sender transmits coded messages, which  reduces the number of broadcast transmissions compared to the naive broadcast of each message. In many practical scenarios, messages are distributed over multiple senders due to data storage limits, or due to erroneous reception of some messages. It can also be  strategically done to reduce end-to-end latency in content delivery, as in multi-server coded caching \cite{MSCC}. Content can be delivered using large storage capacity nodes called caching helpers in cellular networks \cite{KAC}. Data is also distributed and stored over multiple nodes in distributed storage networks \cite{luo2016coded}. Hence, multi-sender index coding is an important component of all the above mentioned problems, which exploit the knowledge of receivers' side-information, when the messages are distributed over multiple senders.

A special class of multi-sender  ICPs was first studied in \cite{SUOH}, where each receiver knows a unique message and demands a subset of other messages. Inner and outer bounds for the capacity region of many variations of multi-sender ICP were provided in \cite{sadeghi2016distributed,YPFK,MOJ2,MOJ}. Existance of links with fixed finite capacities from every sender to every receiver has been assumed in contrast to the previous works. Variations of random coding were used to provide the bounds. A fundamental class of multi-sender ICP is the two-sender ICP, which was first studied in \cite{COJ}. Some single-sender index coding schemes based on graph theory were extended to the two-sender unicast ICP (TUICP), where each receiver demands a unique message. Optimal broadcast rates of a special class of TUICPs, and related code constructions using optimal codes of single-sender sub-problems were studied in \cite{CTLO}, \cite{CVBSR2}, \cite{CVBSR1}. A more practical version of the two-sender ICP is the secure ICP, which arises in on-demand content delivery. The servers do not want a particular client (or a set of clients) to gain any information about any content it does not subscribe for. The client gains access to the unsubscribed content by listening to the transmissions of all the senders, which were  intended for the remaining clients. 
Many variations of the single-sender secure ICP have been studied previously.

Single-sender ICP with security was first studied in \cite{DSC2}. An eavesdropper has a subset of messages and the transmitted codeword. The objective of secure index coding is to encode the messages such that the eavesdropper is unable to gain information about a specified subset of messages it does not have. Different levels of security were introduced. In weak security, the eavesdropper is not able to obtain additional information about each message he does not have. Block security was introduced which generalizes the notion of weak security. Necessary and sufficient conditions for a linear code to be block secure were given. Some relations between the minimum distance of the code, level of block security attained, and the amount of side-information with the eavesdropper were established. Strongly secure index coding was considered in \cite{MAG}, where the eavesdropper does not have any side-information, and must not gain any information about  the message set. This involves sharing random keys with only the receivers, and encoding the  messages along with these keys. 
A necessary and sufficient condition for the existence of weakly-secure index codes for any ICP has been given, when the eavesdropper can access any subset of $t$ messages \cite{LVP}. Three cases have been identified, where random keys are not required to attain weak security. An equivalence between secure network coding and secure index coding was  established in \cite{OKVY}. Capacity region of secure index coding is characterised using an outer bound and an inner bound which uses composite coding scheme \cite{LVKS}. In this paper, we study the construction of weakly secure linear index codes for the TUICP.

The key results of this paper are summarized as follows.

\begin{itemize}
	\item We introduce the problem of two-sender unicast weakly secure index coding against an eavesdropper having some side-information. To the best of our knowledge, this is the first work investigating the two-sender unicast index coding problem with weak security.   
	\item We provide a general  construction of weakly secure linear two-sender index codes for any general TUICP, using weakly secure linear single-sender codes of its sub-problems. The constructed codes need not be optimal. However, this is the first work (to the best of our knowledge) where weakly secure linear codes are constructed using those of the sub-problems.
	\item Code-constructions are given for different classes of the TUICP, using codes of the sub-problems. This establishes upper bounds on the optimal codelengths of TUICPs.
	\item Some classes of the TUICP are identified, where the constructed codes are optimal. This result reduces the problem of finding two-sender optimal weakly secure linear index codes to the problem of finding the  corresponding optimal codes of the single-sender sub-problems.
\end{itemize}

\par The remainder of the paper is organized as follows. Section II establishes the problem setup and provides required definitions. Section III provides the main results of the paper. Conclusion of the paper is provided in Section IV. 

\section{Problem Formulation and Definitions}

In this section, we formulate the problem of two-sender unicast index coding with weak security and establish the required notations and definitions used in this paper. We use the notation employed in  \cite{CVBSR1}. 

\begin{figure*}[!htbp]
\begin{center}
	\includegraphics[width=41pc,height=15pc]{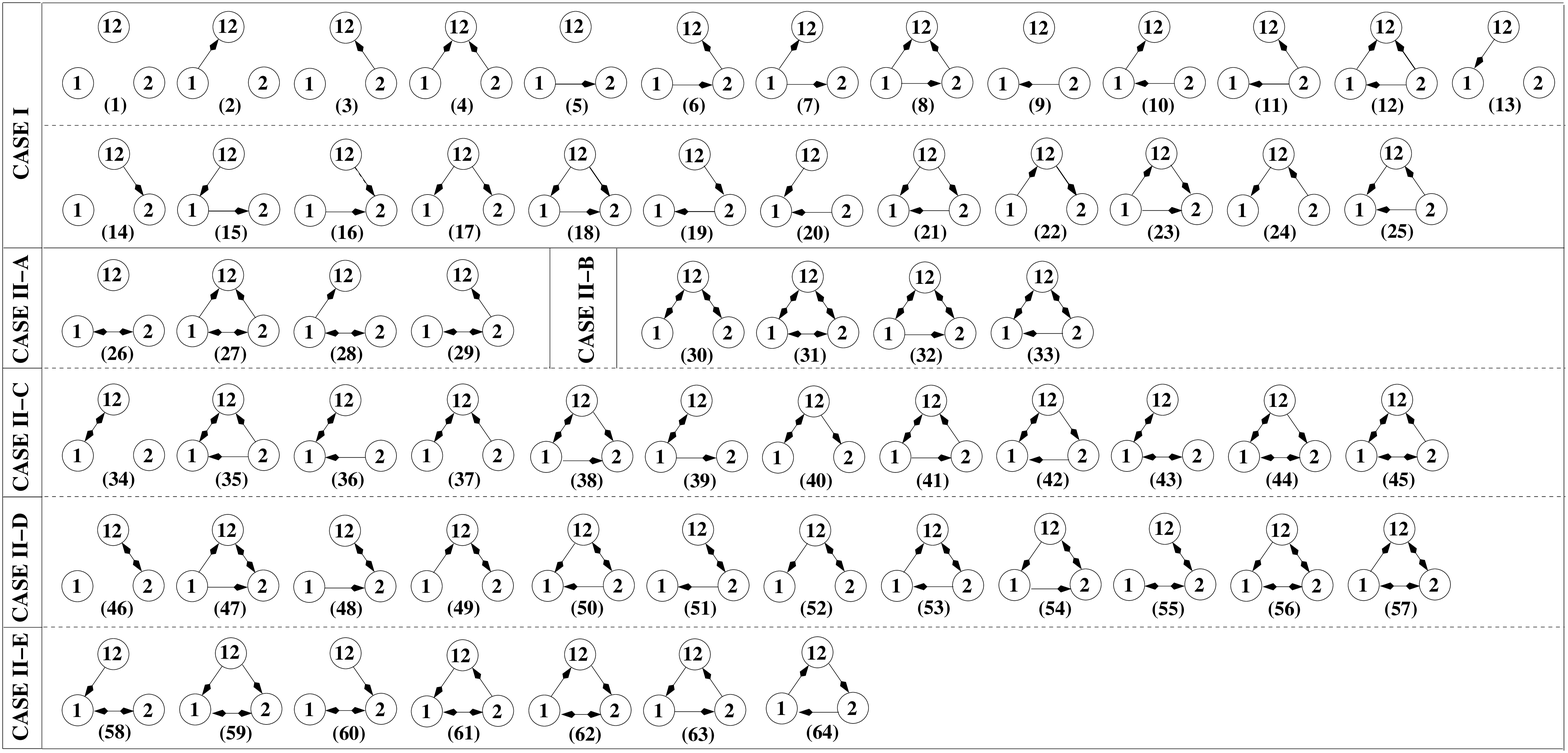}
	\caption{Enumeration of all the possible interactions between the sub-digraphs $\mathcal{D}_1$, $\mathcal{D}_2$, and $\mathcal{D}_{\{1,2\}}$, given by the digraph $\mathcal{H}$.}
	\label{interenum}
\end{center}
\hrule
\end{figure*}

The set $\{1,2,\cdots, n\}$ is denoted as $[n]$. $\mathbb{F}_q$ denotes a finite field of $q$ elements. Support of a vector ${\bf{u}} \in \mathbb{F}_q^m$, denoted as $supp({\bf{u}})$, is the set of co-ordinates where its components are non-zero. For any matrix ${\bf{G}}$, let $\langle {\bf{G}} \rangle$ denote the row space of ${\bf{G}}$. A unicast ICP is an ICP where every receiver requests a unique message. An instance of the two-sender unicast index coding problem (TUICP) consists of two senders collectively having the message set $\mathcal{M} =\{x_1,x_2,\cdots,x_m\}$, where $x_i \in \mathbb{F}_{q}$, $\forall i \in [m]$. Each message $x_i$ is a realization of a random variable $X_i$, $i \in [m]$. The random variables $\{X_i\}_{i \in [m]}$ are assumed to be independent and uniformly distributed over $\mathbb{F}_q$. The $j$th sender denoted by $\mathcal{S}_{j}$, $j \in \{1,2\}$, has a subset of $\mathcal{M}$ given by $\mathcal{M}_j$, such that $\mathcal{M}_1 \cup \mathcal{M}_2 = \mathcal{M}$.  Each sender knows the identity of all the messages present with the other. We assume a noiseless broadcast channel. Transmissions of each sender are orthogonal in time with those of the other. There are $m$ receivers, the $i$th receiver wanting $x_i$ and having $\mathcal{K}_i \subseteq \mathcal{M} \setminus \{x_i\}$, $i \in [m]$, as its side-information.  If $\mathcal{M}_{1}=\mathcal{M}$ and  $\mathcal{M}_{2}=\Phi$, or vice-versa, the TUICP reduces to the single-sender unicast ICP (SUICP).

For any unicast ICP, the knowledge of  side-information and demands of all the receivers is represented by the side-information digraph  $\mathcal{D}=(\mathcal{V}(\mathcal{D}),\mathcal{E}(\mathcal{D}))$, with the vertex set  $\mathcal{V}(\mathcal{D})=\{v_{1},\cdots,v_{m}\}$. The vertex $v_{i}$ represents $i$th receiver which wants $x_i$, $i \in [m]$. The edge set is given by $\mathcal{E}(\mathcal{D})=\{(v_{i},v_{j}): x_{j} \in \mathcal{K}_{i}, i,j \in [m]\}$. Any SUICP can also be represented using a fitting matrix \cite{PK}. It contains unknown entries denoted by $x$. Each row represents a receiver and  each column represents a message.

\begin{defn}[Fitting Matrix, \cite{PK}]
	An $m \times m$ matrix ${\bf{F}}_x$ is called the fitting matrix of an SUICP described by the side-information digraph $\mathcal{D}$, where the ($i,j$)th entry is given by `$x$' if  ${\bf{x}}_j \in \mathcal{K}_i$, $1$ if $i = j$, and $0$ otherwise, for all $i,j \in [m]$.
\end{defn}

The TUICP has been analyzed using three sub-digraphs (equivalently sub-problems) induced by three disjoint vertex sets of the side-information digraph \cite{CTLO},\cite{CVBSR1}. We recapitulate the same, as we employ the same approach. Let $\mathcal{P}_{1} = \mathcal{M}_{1} \setminus \mathcal{M}_{2}$ and $\mathcal{P}_{2} = \mathcal{M}_{2} \setminus \mathcal{M}_{1}$ denote the message sets available only with $\mathcal{S}_{1}$ and $\mathcal{S}_{2}$ respectively. The messages known to both the senders are given by   $\mathcal{P}_{\{1,2\}} = \mathcal{M}_{1} \cap \mathcal{M}_{2}$. Let $m_{\mathcal{T}}=|\mathcal{P}_{\mathcal{T}}|$, for any non-empty set $\mathcal{T} \subseteq  \{1,2\}$. Any  singleton set is represented without $\{\}$. Let $\mathcal{P}=(\mathcal{P}_{1},\mathcal{P}_{2},\mathcal{P}_{\{1,2\}})$. Any TUICP $\mathcal{I}$ can now be described in terms of the two tuple $(\mathcal{D},\mathcal{P})$, as $\mathcal{I}(\mathcal{D},\mathcal{P})$.

 Let $\mathcal{D}_{\mathcal{T}}$ be the sub-digraph of $\mathcal{D}$, induced by the vertices $\{v_j: x_j \in \mathcal{P}_{\mathcal{T}}, j \in [m]\}$, for any non-empty set $\mathcal{T} \subseteq \{1,2\}$. If there exists an edge from some vertex in $\mathcal{V}(\mathcal{D}_{\mathcal{T}})$ to some vertex in $\mathcal{V}(\mathcal{D}_{\mathcal{T}'})$, in the digraph $\mathcal{D}$, for non-empty  $\mathcal{T},\mathcal{T}' \subseteq \{1,2\}, \mathcal{T} \neq \mathcal{T}'$, an interaction is said to exist from $\mathcal{D}_{\mathcal{T}}$ to $\mathcal{D}_{\mathcal{T}'}$, denoted as $\mathcal{D}_{\mathcal{T}} \rightarrow \mathcal{D}_{\mathcal{T}'}$. The interaction $\mathcal{D}_{\mathcal{T}} \rightarrow \mathcal{D}_{\mathcal{T}'}$ is said to be fully-participated, if there are edges from every vertex in $\mathcal{V}(\mathcal{D}_{\mathcal{T}})$ to every vertex in $\mathcal{V}(\mathcal{D}_{\mathcal{T}'})$. Otherwise, it is said to be partially-participated. The TUICP is said to have  fully-participated interactions if all the existing interactions are fully-participated. Consider the digraph $\mathcal{H}$ with  $\mathcal{V}(\mathcal{H}) = \{1,2,\{1,2\}\}$ and $\mathcal{E}(\mathcal{H})=\{(\mathcal{T},\mathcal{T}') | \mathcal{D}_{\mathcal{T}} \rightarrow \mathcal{D}_{\mathcal{T}'}, \mathcal{T},\mathcal{T}' \in \mathcal{V}(\mathcal{H})\}$. We call the digraph $\mathcal{H}$ as the interaction digraph of the digraph $\mathcal{D}$, for a given $\mathcal{P}$. Bidirectional  edges in interaction digraph are denoted by a single edge with arrows at both the ends. There are 64 possibile interaction digraphs given in Figure \ref{interenum}, which were enlisted  and classified in \cite{CTLO}. The vertex representing the set $\{1,2\}$ is written as $12$ for brevity. 
 Note that all the possible interaction digraphs are classified into two cases broadly: Case I and Case II. Case I consists of all acyclic digraphs. Case II was further classified into five subcases as shown in Figure \ref{interenum}.
 
For a given instance of the TUICP, a two-sender index code consists of two sub-codes transmitted by the two senders respectively. Let  ${\bf{x}}^{(\mathcal{T})}$ be the concatenated message vector of messages in $\mathcal{P}_{\mathcal{T}}$, for non-empty $\mathcal{T} \subseteq  \{1,2\}$. The random vectors ${\bf{X}}^{\mathcal{T}}$ are defined similarly. An encoding function for  $\mathcal{S}_{j}$ is given by $\mathbb{E}_{j}:\mathbb{F}_{q}^{|\mathcal{M}_{j}| \times 1} \rightarrow   \mathbb{F}_{q}^{l_{j} \times 1}$, such that ${\bf{c}}_j=\mathbb{E}_{j}({\bf{x}}^{(j)},{\bf{x}}^{(\{1,2\})})$, where $l_j$ is the length of the codeword ${\bf{c}}_j$, $j \in \{1,2\}$. The $i$th receiver has a decoding function given by $\mathbb{D}_{i}:\mathbb{F}_{q}^{(l_{1}+l_{2}+|\mathcal{K}_{i}|) \times 1} \rightarrow   \mathbb{F}_{q}$, such that $x_i = \mathbb{D}_{i}({\bf{c}}_1,{\bf{c}}_2,\mathcal{K}_i)$, $i \in [m]$, i.e., it can decode $x_i$ using its side-information and the received codewords ${\bf{c}}_1$ and ${\bf{c}}_2$.
An index code for a two-sender problem is said to be linear, if both the encoding functions are linear transformations. 
In general, a linear code seen by any receiver can be written as in (\ref{genrec}), where ${\bf{x}}=(x_1, x_2, \cdots, x_m)^T \in \mathbb{F}^{m \times 1}_q$ is the concatenated message vector. The matrix ${\bf{G}}$ is an $\l \times m$ matrix, where $l = l_1+l_2$ is the codelength of the two-sender index code.
\begin{equation}
{\bf{G}}{\bf{x}}=
\left(
\begin{array}{c|c|c}
{\bf{G}}_{1} & {\bf{0}}   & {\bf{G}}_{\{1,2\}}^{1} \\ 
\hline
{\bf{0}} & {\bf{G}}_{2}  & {\bf{G}}_{\{1,2\}}^{2} \\
\end{array}
\right)
\left(
\begin{array}{c}
{\bf{x}}^{(1)}  \\
\hline
{\bf{x}}^{(2)}  \\
\hline
{\bf{x}}^{(\{1,2\})}
\end{array}
\right).
\label{genrec}
\end{equation}
${\bf{G}}$ is also called an encoding matrix of the problem. The constituent matrices have appropriate sizes as seen from the partition of ${\bf{G}}$ in (\ref{genrec}). $\mathcal{S}_i$ sends ${\bf{G}}_i{\bf{x}}^{(i)}+{\bf{G}}_{\{1,2\}}^{i}{\bf{x}}^{(\{1,2\})}$, $i \in \{1,2\}$. Depending on the two-sender problem, any of the matrices  ${\bf{G}}_{\{1,2\}}^{i}$, $i \in \{1,2\}$, can be ${\bf{0}}$. 

We now consider the notion of weak security in the TUICP setup. For any $\mathcal{B} \subset [m]$, let $x_{B}=\{x_i : i \in \mathcal{B}\}$. There is an eavesdropper having a subset of messages $x_{\mathcal{A}}$, with $\mathcal{A} \subsetneq [m]$, and having access to the codewords transmitted by the two senders. The set $\mathcal{A}$ is an element of the set $\mathcal{U} \subset 2^{[m]}$ ($2^{[m]}$ is the set of all subsets of [m]). The set $\mathcal{U}$ consists of sets of indices of possibly compromised messages. The senders know the set $\mathcal{U}$, but not the particular set $\mathcal{A}$. That is, the senders do not know the particular $x_{\mathcal{A}}$ accessed by the eavesdropper, unless $|\mathcal{U}|=1$.
 
A weakly secure linear index code  must ensure that all the receivers are able to obtain their demands, and also that the eavesdropper does not gain additional information about each message not present in $x_{\mathcal{A}}$. That is, assuming the encoding functions $\{\mathbb{E}_j\}_{j \in \{1,2\}}$ to be linear transformations,  $H(X_i|\mathbb{E}_1({\bf{X}}^{(1)},{\bf{X}}^{(\{1,2\})}),\mathbb{E}_2({\bf{X}}^{(2)},{\bf{X}}^{(\{1,2\})}),{\bf{X}}_{\mathcal{A}})=H(X_i)$, for all $i \in \mathcal{A}^c$. $H(X)$ and  $H(W|{\bf{X}},{\bf{Y}},{\bf{Z}})$ denote shannon entropy and conditional entropy. For a TUICP $\mathcal{I}(\mathcal{D},\mathcal{P})$, consider $\mathcal{A}=\mathcal{A}_1 \cup \mathcal{A}_2 \cup \mathcal{A}_{\{1,2\}}$, such that  $x_{\mathcal{A}_{\mathcal{T}}} = x_{\mathcal{A}} \cap x_{\mathcal{P}_{\mathcal{T}}}$, for non-empty $\mathcal{T} \subseteq \{1,2\}$. That is, $x_{\mathcal{A}_{\mathcal{T}}}$ is the side-information of the eavesdropper present in $\mathcal{P}_{\mathcal{T}}$. We denote this instance of the weakly secure TUICP as $\mathcal{I}(\mathcal{D},\mathcal{P},\mathcal{A})$. Similarly, an SUICP with eavesdropper's side-information given by $\mathcal{A}$ is denoted as $\mathcal{I}(\mathcal{D},\mathcal{A})$. The optimal (minimum) length (over a given field) of a weakly secure linear index code for the problem $\mathcal{I}(\mathcal{D},\mathcal{P},\mathcal{A})$ is denoted as $l^{*}(\mathcal{I}(\mathcal{D},\mathcal{P},\mathcal{A}))$. For $\mathcal{I}(\mathcal{D},\mathcal{A})$, it is denoted as $l^{*}(\mathcal{I}(\mathcal{D},\mathcal{A}))$. When there is no ambiguity and the given problem is understood, we write $l^{*}$ instead of $l^{*}(\mathcal{I}(\mathcal{D},\mathcal{P},\mathcal{A}))$ for brevity. Similarly, the optimal codelength  of the subproblem $l^{*}(\mathcal{I}(\mathcal{D}_{\mathcal{T}},\mathcal{A}_{\mathcal{T}}))$ is denoted as $l^{*}_{\mathcal{T}}$, for non-empty $\mathcal{T} \subseteq \{1,2\}$.
 
 \par The following notations are required for the construction of a two-sender code from single-sender codes.
 Any vector ${\bf{u}}_x$ with a subscript $x$ consists of $1$'s, $0$'s and $x$'s. An $x$ denotes an unknown value, which can be replaced by any element from the given field. A vector ${\bf{u}}$ is said to complete a vector ${\bf{u}}_x$ and denote it as ${\bf{u}} \approx {\bf{u}}_x$, if  ${\bf{u}}$ can be obtained by replacing the $x$'s in ${\bf{u}}_x$ with known values from the given field. A vector ${\bf{e}}_{x_{\mathcal{A}}}^{(i,m)}$ is a $1 \times m$ vector with $1$ in the $i$th co-ordinate, $x$'s in the co-ordinates given by $\mathcal{A}$, and $0$'s in the rest of the co-ordinates. The `$x$'s in ${\bf{e}}_{x_{\mathcal{A}}}^{(i,m)}$ denote unknown values.   Let ${\bf{c}}_1$ and ${\bf{c}}_2$ be  two codewords. The notation ${\bf{c}}_1 + {\bf{c}}_2$ denotes the symbol-wise addition of ${\bf{c}}_1$ and ${\bf{c}}_2$ after zero-padding the shorter message at the least significant positions to match the length of the longer one. The notation ${\bf{c}}[a:b]$ denotes the vector obtained by choosing the code symbols from co-ordinate $a$ to $b$, starting from the most significant position. Similar notation holds for codes $\mathcal{C}_1$ and $\mathcal{C}_2$. 
 
 We now illustrate the definitions and notations introduced in this section with a running example.
  
\begin{exmp}
	Consider the TUICP with $m=4$ messages, where the $i$th receiver demands $x_{i} \in \mathbb{F}_2$, $i \in \{1,2,3,4\}$. Sender $\mathcal{S}_1$ has $\mathcal{M}_1=\{x_1,x_2,x_3\}$. $\mathcal{S}_2$ has $\mathcal{M}_2=\{x_3,x_4\}$. Hence, $\mathcal{P}_1=\{x_1,x_2\}$, $\mathcal{P}_2=x_4$, and $\mathcal{P}_{\{1,2\}}=x_3$. The side-information of each receiver is given as follows: $\mathcal{K}_{1}=x_{2}$, $\mathcal{K}_{2}=x_{1}$, $\mathcal{K}_{3}=\{x_{2},x_{4}\}$, $\mathcal{K}_{4}=\{x_{2},x_3\}$. The side-information digraph  $\mathcal{D}$ and the corresponding interaction digraph  $\mathcal{H}$ are shown in Figure \ref{fig2}.  The vertex-induced sub-digraphs $\mathcal{D}_1$, $\mathcal{D}_2$, and $\mathcal{D}_{\{1,2\}}$ are also shown. Note that only the interactions $\mathcal{D}_{\{1,2\}} \rightarrow \mathcal{D}_{2}$ and $\mathcal{D}_{2} \rightarrow \mathcal{D}_{\{1,2\}}$ are fully-participated. The interaction digraph shown in Figure \ref{fig2} is $\mathcal{H}_{50}$ as given in Figure \ref{interenum}. Let the eavesdropper have message set $\{x_3,x_4\}$. Hence, $\mathcal{A}=\{3,4\}$. and $\mathcal{A}_1=\Phi,\mathcal{A}_2=4,\mathcal{A}_{\{1,2\}}=3$. Consider the code $\mathcal{C}=(x_1+x_2,x_3+x_4)$. It can be easily seen that the eavesdropper cannot decode each of $x_1$ and $x_2$. In Lemma 4.3, \cite{DSC2}, it has been shown that under the assumption of linear encoding, the eavesdropper is not able to gain any information about a particular message iff it is not able to decode it. Hence, the code is weakly secure. Note also that the eavesdropper has some joint information about $x_1$ and $x_2$, even though it is not able to decode either. The encoding matrix ${\bf{G}}$ is shown below.
	\begin{gather*}
	{\bf{G}}=
	\left(
	\begin{array}{cc|c|c}
	1 & 1 & 0 & 0 \\ 
	\hline
	0 & 0 & 1 & 1     
	\end{array}
	\right),
	{\bf{G}}_1=
	\left(
	\begin{array}{cc}
	1 & 1 \\ 
	\end{array}
	\right),	
	{\bf{G}}_2=
	\left(
	\begin{array}{c}
	1  
	\end{array}
	\right),\\	
		{\bf{G}}_{\{1,2\}}^1=
		\left(
		\begin{array}{c}
		0
		\end{array}
		\right),
		{\bf{G}}_{\{1,2\}}^2=
		\left(
		\begin{array}{c}
		1
		\end{array}
		\right).
	\end{gather*}
	\label{exmp1}
\end{exmp}

	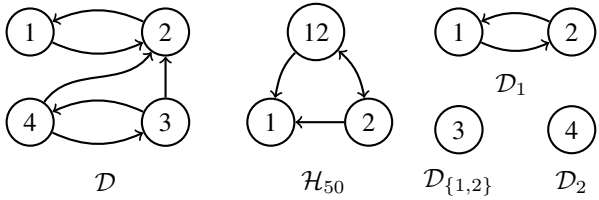
\begin{figure}
		\centering
		\begin{tikzpicture}
		[place/.style={circle,draw=black!100,thick}]
		\node at (-7,0) [place] (1c) {1};
     	\node at (-7,-1.2) [place] (4c) {4};
	 	\node at (-5.2,0) [place] (2c) {2};
	 	\node at (-5.2,-1.2) [place] (3c) {3};	
		\draw (-6,-2) node {$\mathcal{D}$};
     	\draw [thick,->] (1c) to [out=-25,in=-155] (2c);
		\draw [thick,->] (2c) to [out=155,in=25] (1c);
		\draw [thick,->] (3c) to [out=90,in=-90] (2c);
	    \draw [thick,->] (4c) to [out=55,in=-125] (2c);	
        \draw [thick,->] (4c) to [out=-25,in=-155] (3c);
        \draw [thick,->] (3c) to [out=155,in=25] (4c);
		\node at (-3.8,-1.2) [place] (1h) {1};
		\node at (-3.2,0)  [place] (12h) {12};
		\node at (-2.5,-1.2)  [place] (2h) {2};
		\draw [thick,->] (12h) to [out=-135,in=80] (1h);
	   	\draw [thick,<->] (2h) to [out=100,in=-40] (12h);
	   	\draw [thick,->] (2h) to [out=-180,in=0] (1h);
	   	\draw (-3.1,-2) node {$\mathcal{H}_{50}$};
	   	  	\node at (-1.3,-1.3) [place] (3'c) {3};
	   	   	\node at (0.2,-1.3) [place] (4'c) {4}; 
	 		\node at (-1.3,0) [place] (1'c) {1};
	 		\node at (0.2,0) [place] (2'c) {2};  
	     	\draw [thick,->] (1'c) to [out=-25,in=-155] (2'c);
	     	\draw [thick,->] (2'c) to [out=155,in=25] (1'c); 		
	 		\draw (-1.3,-2) node {$\mathcal{D}_{\{1,2\}}$};	
		\draw (0.2,-2) node {$\mathcal{D}_2$};
		\draw (-0.6,-.7) node {$\mathcal{D}_1$};
		\end{tikzpicture}
		\caption{Side-information digraph, interaction digraph, and sub-digraphs of Example \ref{exmp1}}
		\label{fig2}
	\end{figure}
	
\section{Main Results}
In this section, we use weakly secure linear codes of the single-sender sub-problems to obtain a weakly secure linear code for the two-sender problem. The code-constructions given in this section assume the existence of weakly secure linear index codes for all the related sub-problems.

We recapitulate the following lemma given as a special case of Lemma 4.3 in \cite{DSC2}, which gives the necessary and sufficient conditions for a linear index code to be weakly secure. 
\begin{lem}[Lemma 4.3, \cite{DSC2}]
	Let the sender transmit ${\bf{G}}{\bf{x}}$, where ${\bf{G}} \in \mathbb{F}_q^{l \times m}$ and ${\bf{x}} \in \mathbb{F}_q^{m \times 1}$. For each $i \in \mathcal{A}^c$, an eavesdropper having $x_{\mathcal{A}}$, and receiving ${\bf{G}}{\bf{x}}$, has no information about $x_i$ iff the following holds.
	\begin{gather}
	\forall {\bf{u}} : supp({\bf{u}}) \subseteq \mathcal{A},  ~{\bf{u}} + {\bf{e}}_i \notin \langle {\bf{G}} \rangle.
	\label{nesuse}
	\end{gather}
	\label{lemdau} 
\end{lem}

In Lemma \ref{nesusec}, we state a necessary and sufficient condition equivalent to that given in Lemma \ref{lemdau}.

\begin{lem}
	Let an index code be given by the encoding matrix ${\bf{G}} \in \mathbb{F}_q^{l \times m}$, and the eavesdropper have $x_{\mathcal{A}}$, and have access to the index code. For each $i \in \mathcal{A}^c$, an eavesdropper receiving ${\bf{G}}{\bf{x}}$, has no information about $x_i$ iff the following holds.
	\begin{gather}
	\nexists ~~{\bf{d}} \in \mathbb{F}_q^{1 \times l} : {\bf{d}}{\bf{G}} \approx {\bf{e}}_{x_{\mathcal{A}}}^{(i,m)}.
	\label{nesuse1}
	\end{gather}
	\label{nesusec}
\end{lem}
\begin{proof}
	A vector ${\bf{u}}+{\bf{e}}_i$ is in the the row space ${\bf{G}}$, iff there exists a ${\bf{d}}$, such that ${\bf{d}}{\bf{G}}={\bf{u}}+{\bf{e}}_i$. From, the definition of ${\bf{e}}_{x_{\mathcal{A}}}^{(i,m)}$, the condition given in (\ref{nesuse}), is hence equivalent to that given in (\ref{nesuse1}). Hence the result.
\end{proof}

The condition (\ref{nesuse}) is  equivalent to the condition (\ref{nesuse1}).  
The proofs of Lemmas \ref{lem3suicp}, \ref{lem4}, and \ref{a12p12}, and Theorems \ref{thm1} and \ref{thm3} become easier and transparent when condition (\ref{nesuse1}) is employed instead of condition (\ref{nesuse}).
Lemmas \ref{lemdau}, and  \ref{nesusec} are derived for the single-sender ICP. However, they also hold for the TUICP, as the overall transmitted code can be seen in the form of ${\bf{G}}{\bf{x}}$, as in ($\ref{genrec}$). We illustrate Lemma \ref{nesusec} with an example.

\begin{exmp}[Example \ref{exmp1} continued]
Note that for the encoding matrix given in Example \ref{exmp1}, the following are the possible values of ${\bf{d}}{\bf{G}}$: $(1~1~0~0),(0~0~1~1),(1~1~1~1),(0~0~0~0)$. For $\mathcal{A}=\{3,4\}$, $i$ takes  values from  $\mathcal{A}^c=\{1,2\}$. The possible values of ${\bf{e}}^{(i,4)}_{x_{\mathcal{A}}}$ are $(1~0~x~x)$ and $(0~1~x~x)$. With the binary field, the possible values for any vector ${\bf{v}}$ such  that ${\bf{v}} \approx {\bf{e}}^{(1,4)}_{x_{\mathcal{A}}}$ are $(1~0~0~0)$, $(1~0~0~1)$, $(1~0~1~0)$, $(1~0~1~1)$. The possible values for any vector ${\bf{v}}$ such  that ${\bf{v}} \approx {\bf{e}}^{(2,4)}_{x_{\mathcal{A}}}$ are $(0~1~0~0)$, $(0~1~0~1)$, $(0~1~1~0)$, $(0~1~1~1)$. Note that none of the possible values of ${\bf{d}}{\bf{G}}$ are same as any of the eight  vectors which complete ${\bf{e}}^{(i,4)}_{x_{\mathcal{A}}}$, for $i \in \{1,2\}$. Hence, the code obtained is weakly secure.   	
\end{exmp}

The following result shows that a linear weakly secure index code for any SUICP can be obtained by code concatenation of corresponding codes of its subproblems.
 
\begin{lem}
	Consider an SUICP $\mathcal{I}(\mathcal{D},\mathcal{A})$ with the fitting matrix ${\bf{F}}_x$ of size $m \times m$ as given in (\ref{eqlem3suicp}), where ${\bf{F}}_x^{(i)}$ is the $m_i \times m_i$ fitting matrix of an SUICP  $\mathcal{I}(\mathcal{D}_i,\mathcal{A}_i)$, $i \in \{1,2\}$, such that $m_1+m_2=m$. Let the eavesdropper have the message index set $\mathcal{A}=\mathcal{A}_1 \cup \mathcal{A}_2$, with $\mathcal{A}_1 \subset [m_1]$, and $\mathcal{A}_2 \subset [m] \setminus [m_1]$. 
	\begin{equation}
	{\bf{F}}_x=
	\left(
	\begin{array}{c|c}
	{\bf{F}}_{x}^{(1)} & {\bf{B}}_x^{(1)} \\
	\hline
	{\bf{B}}_x^{(2)} & {\bf{F}}_x^{(2)} \\
	\end{array}
	\right)
	\label{eqlem3suicp}
	\end{equation}
	If $l^*_i$ is the optimal codelength, and ${\bf{G}}_{i}$ a corresponding encoding matrix for the problem  $\mathcal{I}(\mathcal{D}_i,\mathcal{A}_i)$, $i \in \{1,2\}$, then the optimal codelength $l^*$ for the problem  $\mathcal{I}(\mathcal{D},\mathcal{A})$ is upper bounded by $l^*_1 + l^*_2$, and a corresponding encoding matrix is given by ${\bf{G}}$ as in (\ref{eq1lem3suicp}).  
	\begin{equation}
	{\bf{G}}=
	\left(
	\begin{array}{c|c}
	{\bf{G}}_{1} & {\bf{0}} \\
	\hline
	{\bf{0}} & {\bf{G}}_{2} \\
	\end{array}
	\right)
	\label{eq1lem3suicp}
	\end{equation}
	\label{lem3suicp}
\end{lem}
\begin{proof}
	We prove the lemma by obtaining a  contradiction. 	From Lemma \ref{nesusec}, we know that 
	\begin{equation}
	\nexists ~~{\bf{d}}_i \in \mathbb{F}_q^{1 \times l_i} : {\bf{d}}_i{\bf{G}}_{i} \approx {\bf{e}}_{x_{\mathcal{A}_i}}^{(j_i,m_i)}, \forall j_i \in [m_i], \forall i \in \{1,2\}.
	\label{eq2lem3suicp} 
	\end{equation}	
	Using Lemma \ref{nesusec}, if ${\bf{G}}$ is not resulting in a weakly secure code, then we have
	\begin{equation*} 
	\exists ~~({\bf{d}}_1|{\bf{d}}_2) \in \mathbb{F}_q^{1 \times (l_1+l_2)} : ({\bf{d}}_1{\bf{G}}_{1}|{\bf{d}}_2{\bf{G}}_{2}) \approx {\bf{e}}_{x_{\mathcal{A}}}^{(j,m)}.
	\end{equation*}
	This implies that (\ref{eq2lem3suicp}) is violated for atleast one of $i \in \{1,2\}$, which contradicts that atleast one ${\bf{G}}_{i}$ will result in a weakly secure code for the problem $\mathcal{I}(\mathcal{D}_i,\mathcal{A}_i)$.
\end{proof}

\subsection{A General Code-construction for any TUICP}

We now provide a general construction of weakly secure linear codes for any TUICP in terms of weakly secure linear codes of its single-sender sub-problems. A weakly secure linear code need not exist for a given SUICP. Hence, the construction is valid only if such codes exist. The two-sender problem can also have partially-participated interactions.

\begin{thm}
Consider any TUICP $\mathcal{I}(\mathcal{D},\mathcal{P},\mathcal{A})$ with any type of interactions. Let $\mathcal{P}_{\{1,2\}}$ be partitioned into two disjoint sets: $\mathcal{P}_{\{1,2\}}^{(1)}$ and $\mathcal{P}_{\{1,2\}}^{(2)}$. The induced sub-digraph of $\mathcal{D}$ induced by the message set  $\tilde{\mathcal{P}}_{i} =\mathcal{P}_{\{1,2\}}^{(i)} \cup \mathcal{P}_{i}$, is denoted by $\tilde{\mathcal{D}}_{i}$, $i \in \{1,2\}$. Let the eavesdropper have side-information $x_{\mathcal{A}}=x_{\tilde{\mathcal{A}}_1} \cup x_{\tilde{\mathcal{A}}_2}$, where $x_{\tilde{\mathcal{A}_i}}=x_{\mathcal{A}} \cap  \tilde{\mathcal{P}}_{i}$, $i \in \{1,2\}$. If there exists a weakly secure linear index code $\tilde{\mathcal{C}}_i$ for each SUICP $\mathcal{I}(\tilde{\mathcal{D}}_i,\tilde{{\mathcal{A}}}_i)$, $i \in \{1,2\}$, then the code $\mathcal{C}=(\tilde{\mathcal{C}}_1,\tilde{\mathcal{C}}_2)$ is a weakly secure linear index code for the TUICP $\mathcal{I}(\mathcal{D},\mathcal{P},\mathcal{A})$.
\label{thm1}
\end{thm}
\begin{proof}
The code $\tilde{\mathcal{C}}_i$
can be written as ${\tilde{\bf{G}}}_i\tilde{{\bf{x}}}_i$, $i \in \{1,2\}$, where ${\tilde{\bf{G}}}_i$ is an $l_i \times |\tilde{\mathcal{P}}_{i}|$ matrix, with $\tilde{{\bf{x}}}_i$ being the concatenated message vector of messages in ${\tilde{\mathcal{P}_i}}$.  Consider the code ($\tilde{\mathcal{C}}_1,\tilde{\mathcal{C}}_2$) given in (\ref{eq1thm1}), obtained from the given codes.
\begin{equation}
{\bf{G}}{\bf{x}}=
\left(
\begin{array}{c|c}
{\tilde{\bf{G}}}_{1} & {\bf{0}} \\
\hline
{\bf{0}} & \tilde{{\bf{G}}}_{2} \\
\end{array}
\right)
\left(
\begin{array}{c}
\tilde{{\bf{x}}}_{1}  \\
\hline
\tilde{{\bf{x}}}_{2}  \\
\end{array}
\right)
\label{eq1thm1}
\end{equation}
All the receivers are able to decode their demands. We now prove that the code is also weakly secure against the eavesdropper by contradiction. Assume that there exists a vector ${\bf{d}}=({\bf{d}}_1|{\bf{d}}_2) \in \mathbb{F}_q^{1 \times l}$, with ${\bf{d}}_i \in \mathbb{F}_q^{1 \times l_i}$, $i \in \{1,2\}$, such that 
${\bf{d}}{\bf{G}} = ({\bf{d}}_1\tilde{{\bf{G}}}_1 | {\bf{d}}_2\tilde{{\bf{G}}}_2) \approx {\bf{e}}_{x_{\mathcal{A}}}^{(j,m)}$, for some $j \in \mathcal{A}^{c}$. If $j \leq |\tilde{\mathcal{P}}_1|$, then we have ${\bf{d}}_1\tilde{{\bf{G}}}_1 \approx {\bf{e}}_{x_{\tilde{\mathcal{A}}_1}}^{(j,|\tilde{\mathcal{P}}_1|)}$. Otherwise we have, ${\bf{d}}_2\tilde{{\bf{G}}}_2 \approx {\bf{e}}_{x_{\tilde{\mathcal{A}}_2}}^{(j-|\tilde{\mathcal{P}}_1|,|\tilde{\mathcal{P}}_2|)}$. In both the cases, this leads to a contradiction as there must not exist such ${\bf{d}}_1$ (or ${\bf{d}}_2$ if $j > |\tilde{\mathcal{P}}_1|$), due to our assumption of the codes given by the matrices $\{\tilde{{\bf{G}}}_i\}_{i \in \{1,2\}}$ being weakly secure, according to Lemma \ref{nesusec}. Hence the result.
\end{proof}

We illustrate the theorem with an example.

\begin{exmp}
Consider the following TUICP with $\mathcal{M}_1=\{x_1,x_2,x_3,x_4,x_5,x_6\}$, and 
$\mathcal{M}_2=\{x_5,x_6,x_7,x_8,x_9\}$. Hence, $\mathcal{P}_1=\{x_1,x_2,x_3,x_4\}$, $\mathcal{P}_2=\{x_7,x_8,x_9\}$, and $\mathcal{P}_{\{1,2\}}=\{x_5,x_6\}$. The $i$th receiver demands $x_i$, $i \in \{1,2,\cdots,9\}$. The side-information of each receiver is: 
\begin{gather*}
\mathcal{K}_1=\{x_2,x_3\},\mathcal{K}_2=\{x_3,x_4\}, \mathcal{K}_3 = \{x_4,x_5\},\\
\mathcal{K}_4=\{x_5,x_1\}, \mathcal{K}_5=\{x_1,x_2\},\mathcal{K}_6=\{x_7,x_8\},\\ \mathcal{K}_7 = \{x_8,x_9\},\mathcal{K}_8=\{x_9,x_6\}, \mathcal{K}_9=\{x_6,x_7\},
\end{gather*}
The eavesdropper has side-information given by $x_{\mathcal{A}}=\{x_2,x_5,x_6,x_8\}$. Hence, $\mathcal{A}_1=\{2\},\mathcal{A}_2=\{8\},\mathcal{A}_{\{1,2\}}=\{5,6\}$. Consider $\mathcal{P}_{\{1,2\}}^{(1)}=x_5$, and $\mathcal{P}_{\{1,2\}}^{(1)}=x_6$. Hence, $\tilde{\mathcal{P}}_1=\{x_1,x_2,x_3,x_4,x_5\}$, and $\tilde{\mathcal{A}_1}=\{2,5\}$. It can be easily verified that the following are valid codes for $\mathcal{I}(\tilde{\mathcal{D}_i},\tilde{\mathcal{A}_i})$, $i \in \{1,2\}$. 
\begin{gather*}
\tilde{\mathcal{C}}_1 = (x_1+x_2+x_3,~x_2+x_3+x_4,~x_3+x_4+x_5),\\
\tilde{\mathcal{C}}_2 = (x_6+x_8,~x_7+x_9).
\end{gather*} 
\end{exmp}

We now state a result for any TUICP (for the sake of  completeness) with any type of interactions, which also holds for Cases I and II-A, without proof. The proof follows on the same lines as that of Theorem \ref{thm1}.
 
\begin{cor}[Naive scheme]
	Consider any TUICP $\mathcal{I}(\mathcal{D},\mathcal{P},\mathcal{A})$ with any type of interactions. If there exists a weakly secure linear code $\mathcal{C}_{\mathcal{T}}$, for the problem $\mathcal{I}(\mathcal{D},\mathcal{A}_{\mathcal{T}})$, for all non-empty $\mathcal{T} \subseteq \{1,2\}$, then the linear code $\mathcal{C}=(\mathcal{C}_1,\mathcal{C}_2,\mathcal{C}_{\{1,2\}})$ is weakly secure for the TUICP $\mathcal{I}(\mathcal{D},\mathcal{P},\mathcal{A})$.
\end{cor}  
	
\begin{rem}
The weakly secure linear index code given in Theorem \ref{thm1} need not be optimal in general.
\end{rem}

\subsection{Code-constructions for some classes of the TUICP belonging to Cases II-B,II-C,II-D, and II-E}

For the remainder of this section, let   $\mathcal{C}_{\mathcal{T}}$ be an optimal weakly secure linear code for the problem  $\mathcal{I}(\mathcal{D}_{\mathcal{T}},\mathcal{A}_{\mathcal{T}})$, for all non-empty $\mathcal{T} \subseteq \{1,2\}$.	
We now state and prove the result for Case II-B.
\begin{thm}[Case II-B]
Consider any TUICP  $\mathcal{I}(\mathcal{D},\mathcal{P},\mathcal{A})$, with the interactions $\mathcal{D}_i  \leftarrow \mathcal{D}_{\{1,2\}}$ and $\mathcal{D}_{\{1,2\}}  \leftarrow \mathcal{D}_{i}$, for all $i \in \{1,2\}$, being fully-participated, and belonging to Case II-B. If weakly secure linear index codes exist for $\mathcal{I}(\mathcal{D}_{\mathcal{T}},\mathcal{A}_{\mathcal{T}})$, for all non-empty $\mathcal{T} \subseteq \{1,2\}$, the optimal codelength  $l^{*}(\mathcal{I}(\mathcal{D},\mathcal{P},\mathcal{A}))$ satisfies $l^{*} \leq max\{l^{*}_{\{1,2\}}, l^{*}_{1}+l^{*}_{2}\}$.	
\label{thm3}
\end{thm}
\begin{proof}
	We consider the code-construction of classical index codes (without weak security restriction) given in the proof of Theorem $7$, in \cite{CTLO}. Then, we substitute optimal weakly secure linear codes in the place of classical index codes given in the construction and show that the resulting code is a weakly secure linear code. There are four sub-cases as follows: 
	$(i)$ $l^*_{\{1,2\}} \geq l^{*}_{1}+l^{*}_{2}$, $(ii)$ $l^*_{\{1,2\}} \geq max\{l^{*}_{1},l^{*}_{2}\}$, $(iii)$ $l^*_{\{1,2\}} \leq l^{*}_{1}$, and $(iv)$ $l^*_{\{1,2\}} \leq l^{*}_{2}$. We mention the code-construction presented in \cite{CTLO} for all the sub-cases, but prove the theorem only for sub-case $(i)$. The proofs for other sub-cases follow on similar lines.
	
	Code construction for sub-cases $(i)$ and $(ii)$: $\mathcal{C}_1 + \mathcal{C}_{\{1,2\}}[1:l^*_1]$ transmitted by $\mathcal{S}_1$, and $\mathcal{C}_2 + \mathcal{C}_{\{1,2\}}[l^*_1+1:l^*_{\{1,2\}}]$ transmitted by $\mathcal{S}_2$. Code construction for sub-case $(iii)$: $\mathcal{C}_1 + \mathcal{C}_{\{1,2\}}$ transmitted by $\mathcal{S}_1$, and $\mathcal{C}_2$ transmitted by $\mathcal{S}_2$.  Code construction for sub-case $(iv)$: $\mathcal{C}_2 + \mathcal{C}_{\{1,2\}}$ transmitted by $\mathcal{S}_2$, and $\mathcal{C}_1$ transmitted by $\mathcal{S}_1$.
	
	Proof for sub-case $(i)$: It has been shown that all the receivers are able to decode their demands from the given code in \cite{CTLO}. Consider any encoding matrix ${\bf{G}}_{\mathcal{T}}$ of the code $\mathcal{C}_{\mathcal{T}}$, for all non-empty $\mathcal{T} \subseteq \{1,2\}$. Let ${\bf{M}}^{[a:b]}$, denote the matrix obtained by taking consecutive rows starting from $a$th row to $b$th row of the matrix ${\bf{M}}$.  It can be  easily verified that the overall encoding matrix ${\bf{G}}$ can be written as shown in (\ref{eq1thm3}).
	\begin{equation}
	{\bf{G}}=
	\left(
	\begin{array}{c|c|c}
	{\bf{G}}_{1} & {\bf{0}}   & {\bf{G}}_{\{1,2\}}^{[1:l^*_1]} \\ 
	\hline
	{\bf{0}} & {\bf{G}}_{2}  & {\bf{G}}_{\{1,2\}}^{[l^*_1+1:l^*_1+l^*_2]} \\
	\hline
	{\bf{0}} & {\bf{0}}  & {\bf{G}}_{\{1,2\}}^{[l^*_1+l^*_2+1:l^*_{\{1,2\}}]} \\
	\end{array}
	\right)
	\label{eq1thm3}
	\end{equation}
	It can be easily shown (as shown in the proof of Theorem \ref{thm1}), that there does not exist any vector ${\bf{d}} \in \mathbb{F}_q^{l^*_{\{1,2\}}}$, such that ${\bf{d}}{\bf{G}} \approx {\bf{e}}_{x_{\mathcal{A}}}^{(j,m)}$, for any $j \in \mathcal{A}^c$. If such a vector exists, it leads to the contradiction that atleast one of the codes of the subproblems is not weakly secure linear, according to Lemma \ref{nesusec}. Hence, the code-construction yields a weakly secure linear code for the TUICP. 
\end{proof}

We state the results for Cases II-C, II-D and II-E without proof. The proofs follow on similar lines as that of Theorem \ref{thm3}. The code-construction for Cases II-C and II-D is based on the code-construction of classical index codes (without weak security restriction) given in the proof of Theorem $8$, in \cite{CTLO}.

\begin{cor}[Case II-C]
	Consider any TUICP  $\mathcal{I}(\mathcal{D},\mathcal{P},\mathcal{A})$, with the interactions $\mathcal{D}_1  \leftarrow \mathcal{D}_{\{1,2\}}$ and $\mathcal{D}_{\{1,2\}}  \leftarrow \mathcal{D}_{1}$ being fully-participated, and belonging to Case II-C. If weakly secure linear index codes exist for $\mathcal{I}(\mathcal{D}_{\mathcal{T}},\mathcal{A}_{\mathcal{T}})$, for all non-empty $\mathcal{T} \subseteq \{1,2\}$, the optimal codelength $l^{*} \leq max\{l^{*}_{\{1,2\}},l^{*}_{1}\}+l^{*}_{2}$. The code-construction giving this bound is given as:	$\mathcal{C}_1 + \mathcal{C}_{\{1,2\}}$ transmitted by $\mathcal{S}_1$, and $\mathcal{C}_2$ transmitted by $\mathcal{S}_2$. 
	\label{cor2}
\end{cor}

The result for Case II-D follows from that of Case II-C, by observing that the interaction digraphs of Case II-D are obtained by interchanging the labels $1$ and $2$ of the vertices in the corresponding interaction digraphs of Case II-C. Hence, we have the following result.

\begin{cor}[Case II-D]
	Consider any TUICP  $\mathcal{I}(\mathcal{D},\mathcal{P},\mathcal{A})$, with the interactions $\mathcal{D}_2  \leftarrow \mathcal{D}_{\{1,2\}}$ and $\mathcal{D}_{\{1,2\}}  \leftarrow \mathcal{D}_{2}$ being fully-participated, and belonging to Case II-D. If weakly secure linear index codes exist for $\mathcal{I}(\mathcal{D}_{\mathcal{T}},\mathcal{A}_{\mathcal{T}})$, for all non-empty $\mathcal{T} \subseteq \{1,2\}$, the optimal codelength $l^{*} \leq max\{l^{*}_{\{1,2\}},l^{*}_{2}\}+l^{*}_{1}$. The code-construction giving this bound is given as:	$\mathcal{C}_2 + \mathcal{C}_{\{1,2\}}$ transmitted by $\mathcal{S}_2$, and $\mathcal{C}_1$ transmitted by $\mathcal{S}_1$. 
		\label{cor3}
\end{cor}

The result of the following corollary depends on the code-constructions given in \cite{CVBSR1} (for sub-cases $(i)$ and $(ii)$ given in the following corollary) and \cite{CTLO} (for sub-case $(iii)$).  
 \begin{cor}[Case II-E]
	Consider any TUICP  $\mathcal{I}(\mathcal{D},\mathcal{P},\mathcal{A})$, with fully-participated interactions, and belonging to Case II-E. If weakly secure linear index codes exist for $\mathcal{I}(\mathcal{D}_{\mathcal{T}},\mathcal{A}_{\mathcal{T}})$, for all non-empty $\mathcal{T} \subseteq \{1,2\}$, the optimal codelength $l^{*} \leq max\{l^{*}_{\{1,2\}}+l^{*}_{1},l^{*}_{\{1,2\}}+l^{*}_{2},l^{*}_{2}+l^{*}_{1}\}$. The code-construction giving this bound depends on three sub-cases:
	\begin{gather*}
	(i) ~\mbox{If}~~l^*_{1} \leq min\{l^*_{2},l^*_{\{1,2\}}\},  ~\mbox{send}	 
	 ~\mathcal{C}_1 + \mathcal{C}_{\{1,2\}}[1:l^*_1] ~\mbox{by}~ \mathcal{S}_1,\\	 
	 ~\mathcal{C}_2 + \mathcal{C}_{\{1,2\}}[1:l^*_1] ~\mbox{by}~ \mathcal{S}_2,
	 ~\mathcal{C}_{\{1,2\}}[l^*_1+1:l^*_{\{1,2\}}] ~\mbox{by}~ \mathcal{S}_1 ~\mbox{or}~ \mathcal{S}_2.\\	 
	 (ii) ~\mbox{If}~~l^*_{2} \leq min\{l^*_{1},l^*_{\{1,2\}}\},  ~\mbox{send}~~
	  \mathcal{C}_1 + \mathcal{C}_{\{1,2\}}[1:l^*_2] ~\mbox{by}~ \mathcal{S}_1 ,	\\  
	  \mathcal{C}_2 + \mathcal{C}_{\{1,2\}}[1:l^*_2] ~\mbox{by}~ \mathcal{S}_2,~~  \mathcal{C}_{\{1,2\}}[l^*_2+1:l^*_{\{1,2\}}] ~\mbox{by}~ \mathcal{S}_1 ~\mbox{or}~ \mathcal{S}_2.\\
	  (iii)~ otherwise,  ~\mbox{send}~ \mathcal{C}_1 + \mathcal{C}_{\{1,2\}} ~\mbox{by}~ \mathcal{S}_1,~ \mathcal{C}_2 + \mathcal{C}_{\{1,2\}} ~\mbox{by}~ \mathcal{S}_2. 
	 \end{gather*}
	 \label{thm7}
\end{cor}

\begin{rem}
	Note that the results in this subsection provide non-trivial upper bounds on the optimal weakly secure index codelength using the code construction given in \cite{CVBSR1}.
\end{rem}

\subsection{A necessary condition for optimality of the construction}
In this subsection, we  obtain a necessary condition for the  optimality of the codes constructed in the previous subsection.
 
From Corollary 4.5 given in \cite{DSC2}, we know that any encoding matrix of any SUICP is obtained by replacing all the unknown entries of the fitting matrix of the SUICP with known entries from the field, and then taking any row basis of such a matrix. We state it here for easy reference.

\begin{lem}[Corollary 4.5, \cite{DSC2}]
	An $r \times m$ matrix ${\bf{G}}$  gives a valid index code for a given SUICP iff for all $i \in [m]$, there exists a vector ${\bf{u}}_i \in \mathbb{F}_q^m$, satisfying $(i)$ $supp({\bf{u}}_i) \subseteq \rchi_i$, where $\rchi_i$ is the set of indices of messages in $\mathcal{K}_i$, and $(ii)$	${\bf{u}}_i+{\bf{e}_i} \in \langle {\bf{G}} \rangle$, where ${\bf{e}}_i$ is the standard basis vector in $\mathbb{F}_q^m$ with $1$ in the  $i$th  co-ordinate and $0$'s in other co-ordinates.
	\label{lemdsc2}
\end{lem}

The following lemma provides an optimal weakly secure linear index code for an SUICP, whose side-information digraph can be partitioned into two vertex-disjoint sub-digraphs with no edges from one sub-digraph to the other or vice-versa.

\begin{lem}
	Consider an SUICP $\mathcal{I}(\mathcal{D},\mathcal{A})$ with the fitting matrix ${\bf{F}}_x$ of size $m \times m$ as given in (\ref{eqlem4}), where ${\bf{F}}_x^{(i)}$ is the $m_i \times m_i$ fitting matrix of an SUICP  $\mathcal{I}(\mathcal{D}_i,\mathcal{A}_i)$, $i \in \{1,2\}$, such that $m_1+m_2=m$. Let the eavesdropper have the message index set $\mathcal{A}=\mathcal{A}_1 \cup \mathcal{A}_2$, with $\mathcal{A}_1 \subset [m_1]$, and $\mathcal{A}_2 \subset [m] \setminus [m_1]$. 
	\begin{equation}
	{\bf{F}}_x=
	\left(
	\begin{array}{c|c}
	{\bf{F}}_{x}^{(1)} & {\bf{0}} \\
	\hline
	{\bf{0}} & {\bf{F}}_x^{(2)} \\
	\end{array}
	\right)
	\label{eqlem4}
	\end{equation}
	If $l^*_i$ is the optimal codelength, and ${\bf{G}}_{i}$ a corresponding encoding matrix for the problem  $\mathcal{I}(\mathcal{D}_i,\mathcal{A}_i)$, $i \in \{1,2\}$, then the optimal codelength $l^*$ for the problem  $\mathcal{I}(\mathcal{D},\mathcal{A})$ is $l^*_1 + l^*_2$, and a corresponding encoding matrix is given by ${\bf{G}}$ as in (\ref{eq1lem4}).  
	\begin{equation}
	{\bf{G}}=
	\left(
	\begin{array}{c|c}
	{\bf{G}}_{1} & {\bf{0}} \\
	\hline
	{\bf{0}} & {\bf{G}}_{2} \\
	\end{array}
	\right)
	\label{eq1lem4}
	\end{equation}
	\label{lem4}
\end{lem}
\begin{proof}
	From Lemma \ref{lemdsc2}, we know that any encoding matrix (not necessary optimal) for the problem $\mathcal{I}(\mathcal{D},\mathcal{A})$ must be of the form given in (\ref{eq1lem4}), with $l_i \times m_i$ matrices  ${\bf{G}}_{i}$, $i \in \{1,2\}$. Using Lemma (\ref{nesusec}), for a weakly secure linear code we must have the following for any $j \in \mathcal{A}^c$.
	\begin{equation} 
	\nexists ~~({\bf{d}}_1|{\bf{d}}_2) \in \mathbb{F}_q^{1 \times (l_1+l_2)} : ({\bf{d}}_1{\bf{G}}_{1}|{\bf{d}}_2{\bf{G}}_{2}) \approx {\bf{e}}_{x_{\mathcal{A}}}^{(j,m)}.
	\end{equation}
	This implies (\ref{eq2lem4}), which is the condition for weak security for the individual encoding matrices ${\bf{G}}_{i}$ of problems $\mathcal{I}(\mathcal{D}_i,\mathcal{A}_i)$, $i \in \{1,2\}$, according to Lemma \ref{nesusec}. Hence minimizing $l_1+l_2$ is equivalent to minimizing each $l_i$, $i \in \{1,2\}$. Thus the result.
	\begin{equation}
	\nexists ~~{\bf{d}}_i \in \mathbb{F}_q^{1 \times l_i} : {\bf{d}}_i{\bf{G}}_{i} \approx {\bf{e}}_{x_{\mathcal{A}_i}}^{(j_i,m_i)}, \forall j_i \in [m_i], \forall i \in \{1,2\}.
	\label{eq2lem4} 
	\end{equation}
\end{proof}

We now use Lemma $\ref{lem4}$, to obtain an optimal weakly secure linear code for a special case of the TUICP, which is then used in Theorem \ref{thm71} to show the optimality of the code-constructions given in the previous subsection (Theorem \ref{thm3} and  Corollaries \ref{cor2}-\ref{thm7}).

\begin{lem}
	Consider any TUICP $\mathcal{I}(\mathcal{D},\mathcal{P},\mathcal{A})$ with any type of interactions and $\mathcal{P}_{\{1.2\}}=\Phi$. We have $l^*=l^*_1+l^*_2$.
	\label{lemp12phi}
\end{lem}
\begin{proof}
	Let the senders send the codewords ${\bf{c}}_i$, $i \in \{1,2\}$. As ${\bf{c}_1}$ is only a function of messages in  $\mathcal{P}_1$, the receivers in $\mathcal{D}_1$ can not make use of their side-information present in $\mathcal{D}_2$, and vice-versa. Hence, it is equivalent to a single-sender problem with receivers in $\mathcal{D}_i$ not having any side-information in $\mathcal{D}_{\{1,2\}\setminus i}$, $i \in \{1,2\}$, and the sender having $\mathcal{P}_1 \cup \mathcal{P}_2$. Now using the result of Lemma \ref{lem4}, we obtain the result of this lemma.
\end{proof}

In the following lemma, we obtain a lower bound on the optimal codelength of a special class of the TUICP using that of a sub-problem.  

\begin{lem}
	Consider any TUICP $\mathcal{I}(\mathcal{D},\mathcal{P},\mathcal{A})$ with any type of interactions and $\mathcal{A}_{\{1,2\}}=\mathcal{P}_{\{1,2\}}$. Consider another related problem $\mathcal{I}(\mathcal{D}',\mathcal{P}',\mathcal{A}')$, where $\mathcal{P}'=(\mathcal{P}_1,\mathcal{P}_2,\Phi)$, $\mathcal{A}'=(\mathcal{A}_1,\mathcal{A}_2,\Phi)$, and  $\mathcal{D}'$ is the induced sub-digraph of $\mathcal{D}$ induced by the messages in $\mathcal{P}'$. Then,  $l^{*}(\mathcal{I}(\mathcal{D}',\mathcal{P}',\mathcal{A}')) \leq l^{*}(\mathcal{I}(\mathcal{D},\mathcal{P},\mathcal{A}))$.
	\label{a12p12}
\end{lem}
\begin{proof}
	Without loss of generality, we consider an optimal encoding matrix for $\mathcal{I}(\mathcal{D},\mathcal{P},\mathcal{A})$ to be of the following form.
	\begin{equation}
	{\bf{G}}=
	\left(
	\begin{array}{c|c|c}
	{\bf{G}}_{1} & {\bf{0}}   & {\bf{G}}_{\{1,2\}}^{1} \\ 
	\hline
	{\bf{0}} & {\bf{G}}_{2}  & {\bf{G}}_{\{1,2\}}^{2} \\
	\end{array}
	\right).
	\label{eq1thm2}
	\end{equation}
	From Lemma \ref{nesusec}, we have the following (assuming ${\bf{G}}_i$ has $l_i$ rows, $i \in \{1,2\}$). 
	\begin{gather*} 
	\nexists ~~({\bf{d}}_1|{\bf{d}}_2) \in \mathbb{F}_q^{1 \times (l_1+l_2)} : \\ ({\bf{d}}_1{\bf{G}}_{1}|{\bf{d}}_2{\bf{G}}_{2}|{\bf{d}}_1{\bf{G}}_{\{1,2\}}^1+{\bf{d}}_2{\bf{G}}_{\{1,2\}}^2) \approx {\bf{e}}_{x_{\mathcal{A}}}^{(i,m)}.
	\end{gather*}
	Hence, we have the following.
	\begin{equation}
	\nexists ~~{\bf{d}}_i \in \mathbb{F}_q^{1 \times l_i} : {\bf{d}}_i{\bf{G}}_{i} \approx {\bf{e}}_{x_{\mathcal{A}_i}}^{(j_i,m_i)}, \forall j_i \in [m_i], \forall i \in \{1,2\}.
	\end{equation}	 
	Thus using Lemma \ref{lemdsc2}, the matrix ${\bf{G}}'=\left(
	\begin{array}{c|c}
	{\bf{G}}_{1} & {\bf{0}} \\ 
	\hline
	{\bf{0}} & {\bf{G}}_{2} \\
	\end{array}
	\right)$ is an encoding matrix for the problem $\mathcal{I}(\mathcal{D}',\mathcal{P}',\mathcal{A}')$. Hence the result.
\end{proof}

\begin{rem}
   Note that in the classical index coding problem optimal codelength of any subproblem (of a given problem) is not more than that of the original problem. This is not proved in general for the case of weakly secure index coding problem.   	
\end{rem}

We now identify some cases where the constructed codes in Theorem \ref{thm3} and Corollaries \ref{cor2}-\ref{thm7} are optimal, using the result of Lemmas \ref{lemp12phi} and  \ref{a12p12}.

\begin{thm}
For any TUICP $\mathcal{I}(\mathcal{D},\mathcal{P},\mathcal{A})$ given in Theorem \ref{thm3} and Corollaries \ref{cor2}-\ref{thm7}, with $\mathcal{A}_{\{1,2\}}=\mathcal{P}_{\{1,2\}}$, if $l^* \leq l^*_1 + l^*_2$, then the codes constructed are optimal. 	
\label{thm71}
\end{thm}
\begin{proof}
  From Lemmas \ref{lemp12phi}  and \ref{a12p12}, we see that $l^*_1+l^*_2 \leq l^*(\mathcal{I}(\mathcal{D},\mathcal{P},\mathcal{A})) \leq l^*_1+l^*_2$. Hence, we have the result.
\end{proof}

We illustrate the theorem using an example.

\begin{exmp}
	Consider the following TUICP with $\mathcal{M}_1=\{x_1,x_2,x_3,x_4,x_5,x_6\}$, and 
	$\mathcal{M}_2=\{x_5,x_6,x_7,x_8,x_9,x_{10}\}$. Hence, $\mathcal{P}_1=\{x_1,x_2,x_3,x_4\}$, $\mathcal{P}_2=\{x_7,x_8,x_9,x_{10}\}$, and $\mathcal{P}_{\{1,2\}}=\{x_5,x_6\}$. The $i$th receiver demands $x_i$, $i \in \{1,2,\cdots,10\}$. The side-information of each receiver is: 
	\begin{gather*}
	\mathcal{K}_1=\{x_2,x_3,x_5,x_6\},\mathcal{K}_2=\{x_3,x_4,x_5,x_6\},\\ \mathcal{K}_3 = \{x_4,x_5,x_6\},\mathcal{K}_4=\{x_1,x_5,x_6\},\\ \mathcal{K}_5=\mathcal{K}_6=x_{[10]}\setminus \{x_5,x_6\},\\ \mathcal{K}_7 = \{x_5,x_6,x_8,x_9,x_{10}\},\mathcal{K}_8=\{x_5,x_6,x_9\},\\ \mathcal{K}_9=\{x_5,x_6,x_7,x_{10}\},\mathcal{K}_{10}=\{x_5,x_6,x_7,x_9\}.
	\end{gather*}
	It can be easily verified that this problem belongs to Case II-B. The eavesdropper has side-information given by $x_{\mathcal{A}}=\{x_2,x_5,x_6,x_8,x_9\}$.  Hence, $\mathcal{A}_1=\{2\},\mathcal{A}_2=\{8,9\},\mathcal{A}_{\{1,2\}}=\{5,6\}$.  It can be easily verified that the following are valid optimal codes for $\mathcal{I}(\mathcal{D}_{\mathcal{T}},\mathcal{A}_{\mathcal{T}})$, for non-empty $\mathcal{T} \in \{1,2\}$. 
	\begin{gather*}
	\mathcal{C}_1 = (x_1+x_2+x_3,~x_2+x_3+x_4,~x_3+x_4),\\
	\mathcal{C}_2 = (x_7+x_9+x_{10},~x_8+x_9),~ \mathcal{C}_{\{1,2\}} = (x_5,~x_6).
	\end{gather*} 
	Hence, $l^*_1=3,l^*_2=2$, and $l^*_{\{1,2\}}=2$. All the conditions given in the Theorem \ref{thm71} are satisfied. Hence, the overall optimal code has length $l^*=5$, and is given as follows:
	\begin{gather*}
	\mathcal{C}_1+\mathcal{C}_{\{1,2\}} = \\ (x_1+x_2+x_3+x_5,~x_2+x_3+x_4+x_6,~x_3+x_4),\\
	 	\mathcal{C}_2 = (x_7+x_9+x_{10},~x_8+x_9).
	\end{gather*} 
\end{exmp}

\section{Conclusion}   
Weakly secure linear index codes are constructed for different classes of the TUICP using those of the sub-problems. For some classes of the TUICP, the constructions are proven to give optimal codes. 

\section*{Acknowledgment}
This work was supported partly by the Science and Engineering Research Board (SERB) of Department of Science and Technology (DST), Government of India, through J.C. Bose National Fellowship to B. S. Rajan.

\end{document}